\documentclass{article}

    \PassOptionsToPackage{numbers, compress}{natbib}

\usepackage[preprint]{neurips_2024}

\usepackage[utf8]{inputenc} 
\usepackage[T1]{fontenc}    
\usepackage{hyperref}       
\usepackage{url}            
\usepackage{booktabs}       
\usepackage{amsfonts}       
\usepackage{nicefrac}       
\usepackage{microtype}      
\usepackage{xcolor}         
\usepackage{graphicx}
\usepackage{subfigure}
\usepackage{multirow}
\usepackage{hyperref}
\usepackage{amsmath}
\usepackage{amssymb}
\usepackage{mathtools}
\usepackage{amsthm}
\usepackage{array} 
\usepackage{pgfplots}
\usepgfplotslibrary{groupplots}

\theoremstyle{plain}
\newtheorem{theorem}{Theorem}[section]

\newtheorem{property}[theorem]{Property}

\theoremstyle{definition}

\theoremstyle{remark}

\definecolor{vir0}{RGB}{68,4,90}
\definecolor{vir1}{RGB}{65,62,133}
\definecolor{vir2}{RGB}{48,104,141}
\definecolor{vir3}{RGB}{31,146,139}
\definecolor{vir4}{RGB}{53,183,119}
\definecolor{vir5}{RGB}{145,213,66}
\definecolor{vir6}{RGB}{248,230,32}

\title{Neural Preconditioning Operator for \\Efficient PDE Solves}

\author{
Zhihao Li\\
The Hong Kong University \\of Science and Technology (Guangzhou)\\
\texttt{zli416@connect.hkust-gz.edu.cn}\\
\And
Di Xiao\\
The Hong Kong University \\of Science and Technology (Guangzhou)\\
\texttt{shawd18376025@gmail.com}
\And
Zhilu Lai\\
The Hong Kong University \\of Science and Technology (Guangzhou)\\
\texttt{zhilulai@ust.hk}\\
\And
Wei Wang\thanks{Corresponding author}\\
The Hong Kong University \\of Science and Technology (Guangzhou)\\
\texttt{weiwcs@ust.hk}\\
}

\begin{document}

\maketitle

\begin{abstract}
We introduce the Neural Preconditioning Operator (NPO), a novel approach designed to accelerate Krylov solvers in solving large, sparse linear systems derived from partial differential equations (PDEs). Unlike classical preconditioners that often require extensive tuning and struggle to generalize across different meshes or parameters, NPO employs neural operators trained via condition and residual losses. This framework seamlessly integrates with existing neural network models, serving effectively as a preconditioner to enhance the performance of Krylov subspace methods. Further, by melding algebraic multigrid principles with a transformer-based architecture, NPO significantly reduces iteration counts and runtime for solving Poisson, Diffusion, and Linear Elasticity problems on both uniform and irregular meshes. Our extensive numerical experiments demonstrate that NPO outperforms traditional methods and contemporary neural approaches across various resolutions, ensuring robust convergence even on grids as large as \(4096\), far exceeding its initial training limits. These findings underscore the potential of data-driven preconditioning to transform the computational efficiency of high-dimensional PDE applications.
\end{abstract}

\section{Introduction}
Partial differential equations (PDEs) are widely used to model phenomena such as heat transfer, electrostatics, and quantum mechanics \cite{82:AnIntro,87:Levy,02:PDE,20:Turbulence}, and are typically solved by discretizing them into large, sparse linear systems \(A x = b\) \cite{77:Numerical}. Krylov subspace methods efficiently handle these systems by projecting onto lower-dimensional subspaces \cite{07:Computational}, yet parametric PDEs and high-resolution meshes often produce ill-conditioned matrices with entries spanning multiple scales, posing considerable computational challenges. Preconditioning is therefore critical for accelerating convergence by improving the system’s spectral properties. Classical techniques like Jacobi, Gauss--Seidel, or multigrid \cite{00:tutorial} can substantially reduce iteration counts, but they frequently require problem-specific tuning and may fail to generalize across changing PDE parameters or domain discretizations.

Neural operators \cite{21:deeponet,21:fno,21:Choose} have recently emerged as a powerful tool for mapping parametric inputs to PDE solutions, learning universal operators directly from data. Building on these developments, researchers have begun leveraging neural operator approaches to speed up classical numerical methods by training data-driven preconditioners \cite{19:LearningNeural,23:LearningPre,24:fcg-no}. While this paradigm offers promising adaptability and the potential to handle complex PDE parameters, it also presents new challenges: extensive training data may be required, large-scale systems often demand specialized architectures, and rigorous theoretical guarantees for convergence—especially under changing boundary conditions or mesh refinements—remain an open problem.

Motivated by these challenges, we introduce a new \emph{Neural Preconditioning Operator} (NPO) framework designed to address the difficulties of data collection, loss function design, and large-scale PDE solving. Our approach first generates datasets by discretizing PDEs, then collects solution and residual information under varying system matrices. We define two key losses—\emph{condition loss} and \emph{residual loss}—to guide the neural operator toward learning effective preconditioning strategies for Krylov subspace methods. Building on these foundations, we propose a \emph{Neural Algebraic Multigrid} operator that blends algebraic multigrid principles with a transformer-based architecture. By exploiting the hierarchical structure of multigrid, the NPO accelerates error reduction across multiple scales while retaining the adaptability inherent in data-driven models. As we show through comprehensive experiments, NPO outperforms both classical and existing neural preconditioners, achieving superior convergence rates and generalizing effectively across different PDE domains and resolutions. Our contributions are as follows:

\begin{itemize}
    \item \textbf{Neural Preconditioning Operator (NPO) Framework.} We propose a framework encompassing data generation, neural operator training (via condition and residual losses), and integration with Krylov subspace solvers, ultimately enabling more efficient and robust PDE solves across varying mesh types and dimensions. Additionally, we provide theoretical analysis to show that our method inherits strong convergence guarantees from classical multigrid, ensuring rapid and stable error reduction.

    \item \textbf{Neural Algebraic Multigrid (NAMG) Operator.} By embedding algebraic multigrid principles within a transformer-based network, we accelerate multiscale error elimination while preserving the adaptability of a data-driven model. The theoretical foundations further confirm that the learned operators cluster eigenvalues around 1, thereby improving the system’s spectral properties.

    \item \textbf{Comprehensive Evaluation.} We benchmark NPO on diverse PDEs and mesh resolutions, demonstrating faster convergence and broader generalization compared to classical and existing neural approaches. Our experiments validate that the learned preconditioners consistently reduce iteration counts and maintain robust performance across various problem settings.
\end{itemize}

\section{Preliminaries}
In this section, we introduce the essential concepts and notations for understanding the proposed neural precondition operator. We begin with an overview of Krylov subspace methods for large linear systems, followed by a discussion of preconditioning techniques.

\subsection{Krylov Subspace Methods}
Krylov subspace iterative methods solve large-scale systems of the form
\begin{equation}
    A\mathbf{x} = \mathbf{b},
\end{equation}
where \(A \in \mathbb{R}^{n \times n}\) is sparse, \(\mathbf{x}\in \mathbb{R}^n\) is the solution, and \(\mathbf{b}\in \mathbb{R}^n\) is the right-hand side (RHS). Starting from an initial guess \(\mathbf{x}_0\), the residual is 
\begin{equation}
    \mathbf{r}_0 = \mathbf{b} - A\mathbf{x}_0.
\end{equation}
The \(m\)-th Krylov subspace is
\begin{equation}
    \mathcal{K}_m(A, \mathbf{r}_0) = \text{span}\{\mathbf{r}_0,\, A\mathbf{r}_0,\, A^2\mathbf{r}_0,\, \dots\}.
\end{equation}

Krylov methods (e.g., Conjugate Gradient, GMRES) construct approximate solutions \(\mathbf{x}_m \in \mathcal{K}_m\) that minimize the residual \(\|\mathbf{b} - A\mathbf{x}_m\|\) in some norm. Further details on Krylov Subspace methods can be found in Appendix~\ref{sec:krylov}.

\subsection{Preconditioner}
Preconditioning improves the convergence of Krylov methods by transforming the system \(A\mathbf{x} = \mathbf{b}\) into one with more favorable spectral properties:
\begin{equation}
    M A \mathbf{x} = M \mathbf{b},
\end{equation}
where \(M \in \mathbb{R}^{n \times n}\) approximates the inverse operator \(A^{-1}\) and is easily invertible. An effective preconditioner \(M\) should closely approximate \(A^{-1}\), thereby clustering the eigenvalues of \(M A\) and reducing the system's condition number. Additionally, \(M\) must be computationally efficient to apply, as excessive overhead can negate the benefits of accelerated convergence. The preconditioner should also scale effectively to large problem sizes, maintaining its performance as the dimension of the system increases. Finally, it should remain robust across varying problem instances, consistently providing improvements even when the parameters or structure of the system change. A comprehensive introduction to preconditioning methods, both numerical and neural, is provided in Appendix~\ref{appendix:relate}.

\subsection{Discretization}
\label{subsec:discretization}

To numerically solve PDEs, we discretize the spatial domain \(\Omega\) using finite differences \cite{07:Finite}, finite volumes \cite{00:Finite}, or finite elements \cite{12:Numerical}, thereby converting the continuous system into $A\mathbf{x} = \mathbf{b}$, where \(A\in\mathbb{R}^{n\times n}\) is a sparse, symmetric positive-definite (SPD) stiffness matrix, \(\mathbf{b}\in\mathbb{R}^n\) encodes source terms and boundary conditions, and \(\mathbf{x}\in\mathbb{R}^n\) contains the nodal approximations of the solution. The SPD property \(\mathbf{v}^\top A\,\mathbf{v}>0\) for all nonzero \(\mathbf{v}\) ensures well-posedness and preserves the key operator characteristics needed for stable numerical solutions.

\subsection{Neural Operators}
We employ a neural operator approach based on the FNO framework \cite{21:fno,23:no} to learn mappings between function spaces. Concretely, let \(D \subset \mathbb{R}^d\) be the spatial domain, and consider an operator $\mathcal{G}: \mathcal{A}(D;\mathbb{R}^{d_a}) \to \mathcal{U}(D;\mathbb{R}^{d_u})$, where \(a \in \mathcal{A}(D;\mathbb{R}^{d_a})\) and \(u \in \mathcal{U}(D;\mathbb{R}^{d_u})\) are functions on \(D\). A neural operator \(\mathcal{N}\) approximates \(\mathcal{G}\) via the composition
\begin{equation}
    \mathcal{N}(a) = \mathcal{Q}\circ\mathcal{L}_L\circ\cdots\circ\mathcal{L}_1\circ\mathcal{R}(a),
\end{equation}
where \(\mathcal{R}\) lifts \(a\) to a higher-dimensional representation, \(\{\mathcal{L}_\ell\}\) are neural operator layers, and \(\mathcal{Q}\) projects the final output to the target dimension.

\begin{figure*}
    \centering
    \includegraphics[width=\textwidth]{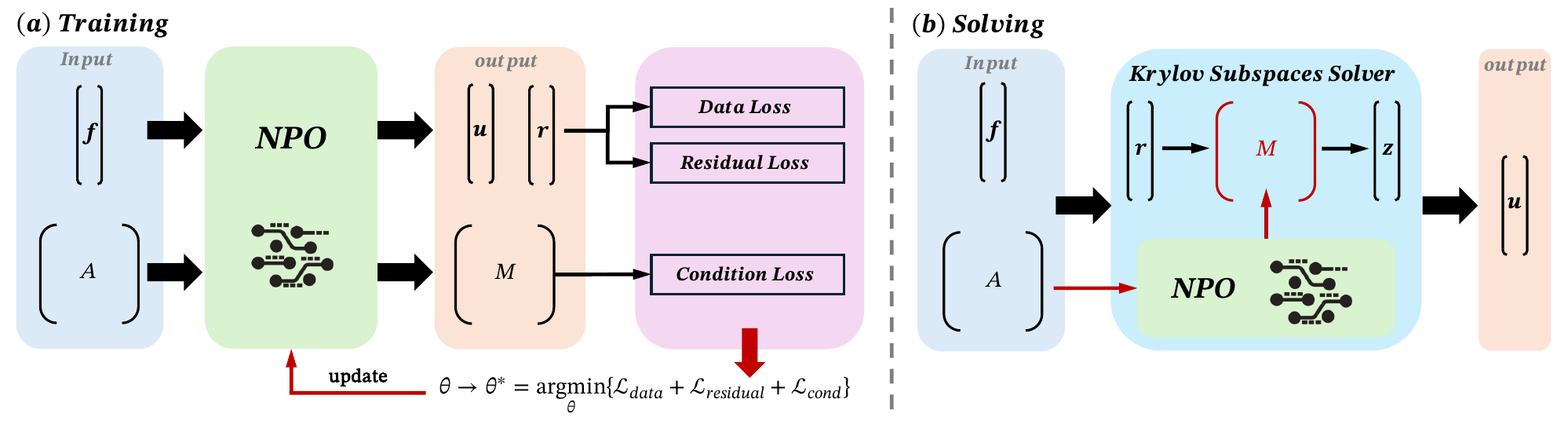}
    \vspace{-10pt}
    \caption{\textbf{Illustration of Neural Preconditioning Operator Framework.} (a) The training phase with multiple loss functions (data, residual, and condition losses) and (b) the solving phase integrated with Krylov subspace methods for efficient PDE solutions.}
    \label{fig:framework}
\end{figure*}

\section{Methodology}
This section presents our neural approach to preconditioning PDEs. We begin by formulating the problem and discretizing the governing PDEs in Section~\ref{subsec:problem_formulation}, followed by an overview of the Neural Preconditioning Operator (NPO) framework in Section~\ref{subsec:npo_framework}. Next, we define the learning objectives for training the NPO in Section~\ref{subsec:learning_npo}, and conclude with a detailed description of the Neural Algebraic Multigrid (NAMG) Operator in Section~\ref{subsec:npo_amg}, which combines classical multigrid principles with neural attention for efficient coarse-grid correction.

\subsection{Problem Formulation}
\label{subsec:problem_formulation}
We consider PDEs on a domain \(D \subset \mathbb{R}^d\), with functions from the input and solution spaces \(\mathcal{A}(D; \mathbb{R}^{d_a})\) and \(\mathcal{U}(D; \mathbb{R}^{d_u})\), respectively. The operator \(\mathcal{G}: \mathcal{A} \to \mathcal{U}\) is expressed as an integral:
\begin{equation}
    \mathcal{G}a(\mathbf{x}) = \int_{D} \kappa(\mathbf{x}, \mathbf{y}) \, a(\mathbf{y}) \, d\mathbf{y},
\end{equation}
where \(\kappa: D \times D \to \mathbb{R}\) is the kernel function.

After discretization, the PDE leads to a sparse, symmetric positive definite (SPD) matrix \(A \in \mathbb{R}^{n \times n}\) and a right-hand side vector \(\mathbf{f} \in \mathbb{R}^n\). Our goal is to learn a preconditioner \(M = \mathcal{M}_{\theta}(A, \mathbf{f})\), defined by:
\begin{equation}
    M \;=\; \mathcal{M}_{\theta}(A),
\end{equation}

where \(\theta\) are the learned parameters. The preconditioner \(M\) is trained to remain SPD and efficient to apply, improving the condition number of \(A\) and accelerating iterative solvers.

\subsection{Neural Preconditioning Operator Framework}
\label{subsec:npo_framework}
Figure~\ref{fig:framework} illustrates the two-phase workflow of our Neural Preconditioning Operator (NPO) framework, consisting of \emph{training} (Figure~\ref{fig:framework}(a)) and \emph{solving} (Figure~\ref{fig:framework}(b)). 

During the training phase, the NPO takes the system matrix \(A\) and right-hand side vector \(f\) as inputs and generates an intermediate output, including a preconditioner matrix \(M\), the solution approximation \(u\), and residual \(r\). Three loss functions are used to guide the optimization: the \emph{data loss} (from \(u\) and \(f\)), \emph{residual loss} (from \(r\)), and \emph{condition loss} (from \(M\)). The NPO's parameters \(\theta\) are updated by minimizing the sum of these losses.

Once trained, the NPO is applied in the solving phase to accelerate iterative Krylov subspace methods (e.g., CG or GMRES). Given a new system \(A\mathbf{x} = \mathbf{b}\), the solver repeatedly uses the learned \(M\) to compute preconditioned residuals \(z = M r\), significantly reducing iteration counts and improving convergence efficiency across various PDE systems and mesh types.

\subsection{Learning Neural Preconditioner Operator}
\label{subsec:learning_npo}
To train a neural preconditioner \( \mathcal{M}_{\theta}(A) \), we define two complementary loss functions: a \emph{condition loss} and a \emph{residual loss}. These losses guide the preconditioner to behave like \( A^{-1} \), improving both the spectral properties of the system and solution accuracy.

\subsubsection{Condition Loss}

A preconditioner that approximates \( A^{-1} \) should ensure that \( A \mathcal{M}_{\theta}(A) \approx I \). A natural objective is to minimize:
\begin{equation}
    \label{eq:inverse_loss}
    \bigl\| I - A\,\mathcal{M}_{\theta}(A) \bigr\|_F^2.
\end{equation}
However, directly optimizing this matrix norm is computationally infeasible for large systems. Instead, we define a condition loss over sampled residuals \(\mathbf{r}_i\) to achieve a similar effect:
\begin{equation}
    \label{eq:condition_loss}
    \min_{\theta} \frac{1}{N} \sum_{i=1}^{N} \bigl\| \bigl(I - A_i\,\mathcal{M}_{\theta}(A_i)\bigr)\,\mathbf{r}_i \bigr\|_2^2.
\end{equation}

This condition loss indirectly improves the system's spectral properties, reducing the condition number of the preconditioned matrix and thereby accelerating convergence in iterative solvers.

\subsubsection{Residual Loss}

While the condition loss ensures better spectral properties, it does not directly assess how well the preconditioner solves the system for the right-hand side \(\mathbf{b}_i\). To address this, we define a residual loss that measures the accuracy of the preconditioner when applied to \(\mathbf{b}_i\):
\begin{equation}
    \label{eq:residual_loss}
    \min_{\theta \in \Theta} 
    \frac{1}{N}
    \sum_{i=1}^{N}
    \bigl\|
       A_{i}\mathcal{M}_{\theta}(A_i)\bigl(\mathbf{b}_i\bigr)
       \;-\;
       \mathbf{b}_i
    \bigr\|_2^2.
\end{equation}

This loss encourages \( \mathcal{M}_{\theta}(A) \) to approximate \( A^{-1} \) by minimizing the discrepancy between the predicted and actual right-hand side. Together, the condition and residual losses promote a preconditioner that reduces both spectral issues and iteration counts, enabling faster and more robust convergence for a wide range of PDE systems.

\subsection{Neural Algebraic Multigrid Operator}
\label{subsec:npo_amg}

The Neural Algebraic Multigrid (NAMG) Operator enhances the classical AMG framework by introducing neural attention mechanisms for efficient feature aggregation and prolongation. The process involves three main steps: restriction, attention-based coarse-grid correction, and prolongation.

\subsubsection{Restriction and Coarse Feature Aggregation}

Given fine-grid features \( \mathbf{x}^{f} \in \mathbb{R}^{N \times C} \) and the adjacency matrix \( A \in \mathbb{R}^{N \times N} \), restriction is defined as:

\begin{equation}
    \mathbf{x}^{c} = R \mathbf{x}^{f}, \quad R = A \cdot E_{\theta},
\end{equation}

where \( E_{\theta} \) contains learned attention weights:

\begin{equation}
    e_{ji} = \frac{\exp(\mathbf{W}_{\text{coarse}} \mathbf{x}_{i}^{f} / \tau)}{\sum_{i' \in \mathcal{N}_j} A_{ji'} \exp(\mathbf{W}_{\text{coarse}} \mathbf{x}_{i'}^{f} / \tau)}.
\end{equation}

Here, \( \mathcal{N}_j \) denotes the neighbors of node \( j \), \( \mathbf{W}_{\text{coarse}} \) is a learnable weight matrix, and \( \tau \) is a scaling parameter. Coarse features are computed by aggregating fine-grid tokens using these weights.

\subsubsection{Attention-Based Coarse Correction}

The coarse-grid features are refined through self-attention. Queries, keys, and values are computed as:

\begin{equation}
    \mathbf{q} = \mathbf{W}_{q} \mathbf{x}^{c}, \quad \mathbf{k} = \mathbf{W}_{k} \mathbf{x}^{c}, \quad \mathbf{v} = \mathbf{W}_{v} \mathbf{x}^{c}.
\end{equation}

Attention scores are then used to update the coarse features:

\begin{equation}
    \mathbf{x}_{j}^{c, \text{updated}} = \sum_{k} \text{softmax}\left( \frac{\mathbf{q}_{j} \cdot \mathbf{k}_{k}^{\top}}{\sqrt{C}} \right) \mathbf{v}_{k}.
\end{equation}

\subsubsection{Prolongation and Fine-Grid Correction}

The updated coarse features are projected back to the fine grid:

\begin{equation}
    \mathbf{x}'^{f} = \mathbf{x}^{f} + P \mathbf{x}'^{c}, \quad P = A \cdot E_{\theta}^{\top}.
\end{equation}

This process dynamically adjusts restriction and prolongation through learned attention, allowing the operator to capture complex patterns inherent in PDEs across diverse domains. 

\begin{figure}
    \centering
    \includegraphics[width=0.6\linewidth]{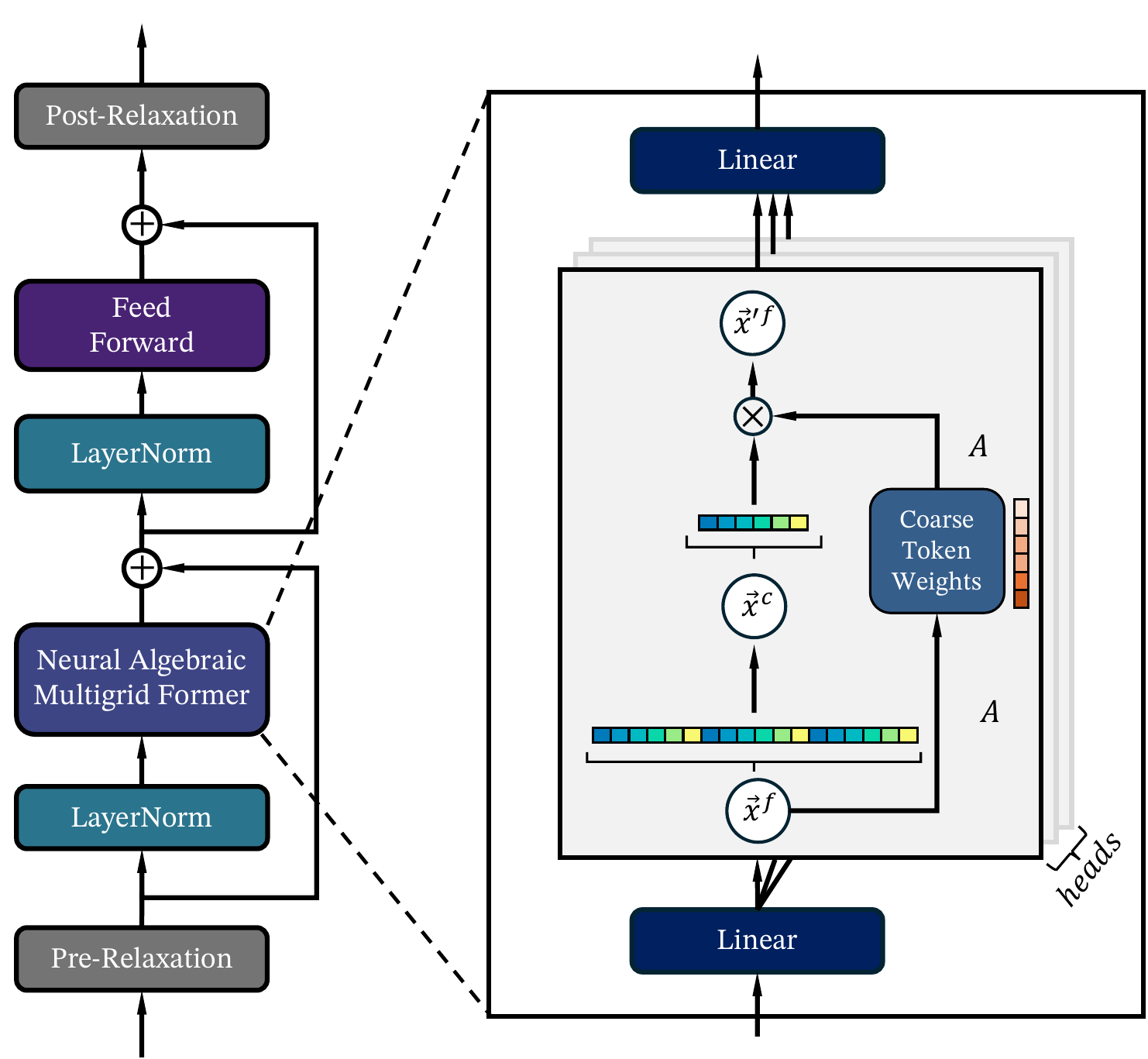}
    \vspace{-10pt}
    \caption{\textbf{Illustration of Neural Algebraic Multigrid Operator.}}
    \label{fig:arch}
\end{figure}

\section{Theoretical Analysis}

\subsection{Convergence Analysis}
We begin by examining the two-grid iteration under standard multigrid assumptions, emphasizing how our approach handles both high- and low-frequency error components. In particular, a smoothing operator \(S\) damps high-frequency errors, while coarse-grid correction addresses low-frequency modes. This two-part strategy yields a convergence rate that does not degrade with increasing problem size.

\begin{theorem}[Two-Grid Convergence]
\label{th:twogrid_convergence}
Let \(\mathbf{e}^{(k)}\) be the error at iteration \(k\) of a two-grid scheme for the SPD system \(A\mathbf{x}=\mathbf{b}\). Suppose the coarse correction satisfies the \emph{Approximation Property} and the smoothing step remains stable. Then there exists a constant \(\rho < 1\) such that
\[
    \|\mathbf{e}^{(k+1)}\|_{a}
    \;\le\;
    \rho\,\|\mathbf{e}^{(k)}\|_{a},
\]
where \(\|\cdot\|_{a}\) is the energy norm induced by \(A\). Consequently, the iteration converges at a rate independent of the system size \(n\).
\end{theorem}

\noindent
(See Appendix~\ref{appendix:proof_2} for proof.) A key ingredient is the coarse space’s ability to capture smooth (low-frequency) errors. In classical multigrid, this is formalized by the \emph{Approximation Property}:

\begin{property}[Approximation Property]
\label{prop:approximation}
Let \(P\in\mathbb{R}^{n\times m}\) be the prolongation operator from the coarse space \(\mathbb{R}^m\) to the fine space \(\mathbb{R}^n\). For any error vector \(\mathbf{e}\in \mathbb{R}^n\), there exists a coarse representation \(\mathbf{z}^c \in \mathbb{R}^m\) such that
\[
    \min_{\mathbf{z}^c}
    \|\mathbf{e} - P\,\mathbf{z}^c\|_{a}
    \;\le\;
    \alpha\,\|\mathbf{e}\|_{a},
\]
with \(\alpha < 1\). (See Appendix~\ref{appendix:proof_4} for details.) 
\end{property}

\noindent
(See Appendix~\ref{appendix:proof_1} for proof.) Together, Theorem~\ref{th:twogrid_convergence} and the Approximation Property ensure uniform error reduction per iteration: smoothing removes high-frequency errors, while the coarse grid approximates low-frequency errors sufficiently well.

\subsection{Operator Properties}

Beyond two-grid convergence, another measure of preconditioner quality is its impact on the spectrum of \(MA\). If \(M\approx A^{-1}\), the eigenvalues of \(MA\) should lie near 1, yielding a small condition number and fast convergence for Krylov solvers like CG or GMRES.

\begin{theorem}[Preconditioned Spectrum Clustering]
\label{th:spectrum_clustering}
Let \(A \in \mathbb{R}^{n \times n}\) be SPD, and let \(M\) be a preconditioner satisfying the smoothing and coarse approximation assumptions of Theorem~\ref{th:twogrid_convergence}. Then there exist constants \(0 < \lambda_{\min} \le \lambda_{\max}\) close to 1 such that every eigenvalue \(\lambda\) of \(MA\) lies in the interval \([\lambda_{\min}, \lambda_{\max}]\). Hence, \(\kappa(MA)=\lambda_{\max}/\lambda_{\min} \) remains near 1, ensuring rapid convergence.
\end{theorem}

\noindent
(See Appendix~\ref{appendix:proof_3} for proof.) This spectral view complements the two-grid analysis by connecting eigenvalue clustering to reduced iteration counts. In a neural setting, the learned operators play a role analogous to restriction, prolongation, and smoothing, thereby preserving these spectral advantages.

Finally, we note that our Neural Algebraic Multigrid (NAMG) Operator can be interpreted as a learnable integral over the domain \(\Omega\), supporting adaptive feature extraction on coarse grids. Formally:

\begin{theorem}[NAMG Operator as a Learnable Integral]
\label{th:integral}
Given an input \(a:\Omega\to\mathbb{R}^d\) and a point \(\mathbf{x}\in\Omega\), the NAMG operator approximates
\[
    \mathcal{G}a(\mathbf{x}) 
    \;=\;
    \int_{\Omega} 
    \kappa(\mathbf{x}, \boldsymbol{\xi})\,a(\boldsymbol{\xi})
    \,d\boldsymbol{\xi},
\]
for some learnable kernel \(\kappa\). 
\end{theorem}

\noindent
(See Appendix~\ref{appendix:proof_4} for proof.) This perspective unifies the classical AMG principle of coarse-grid correction with a data-driven, integral-based formulation. It underscores the capacity of neural operators to adaptively handle diverse PDE structures, ultimately enhancing both convergence and generalization.

\begin{table*}[ht]
\centering
\caption{Performance comparison of GMRES with different methods on Poisson equation with different mesh type, showing time and iterations at different tolerances.}
\vspace{5pt}
\label{table:poisson}
\renewcommand\arraystretch{1.0}
\begin{sc}
    \renewcommand{\multirowsetup}{\centering}
    \resizebox{\linewidth}{!}{
    \begin{tabular}{l|cccc|cccc|cccc}
        \toprule
        Mesh Type & \multicolumn{4}{c}{Grid=512} & \multicolumn{4}{c}{Grid=32*32} & \multicolumn{4}{c}{Irregular} \\
        \cmidrule(lr){1-1} \cmidrule(lr){2-5} \cmidrule(lr){6-9} \cmidrule(lr){10-13} \multirow{2}{*}{Method} & Time (s) & \multicolumn{3}{c}{Iteration} & Time (s) & \multicolumn{3}{c}{Iteration} & Time (s) & \multicolumn{3}{c}{Iteration} \\
        \cmidrule(lr){2-2} \cmidrule(lr){3-5} \cmidrule(lr){6-6} \cmidrule(lr){7-9} \cmidrule(lr){10-10} \cmidrule(lr){11-13} & 1E-10 & 1E-10 & 1E-6 & 1E-4 & 1E-10 & 1E-10 & 1E-6 & 1E-4 & 1E-10 & 1E-10 & 1E-6 & 1E-4 \\
        \midrule
        Jacobi & 3.719 & 513 & 513 & 510 & 2.7337 & 113 & 84 & 62 & 0.455 & 153 & 107 & 73 \\
        Gauss Seidel & 3.536 & 502 & 495 & 494 & 1.7569 & 81 & 68 & 57 & 0.379 & 130 & 99 & 76 \\
        SOR & 3.452 & 502 & 495 & 494 & 2.2438 & 81 & 68 & 57 & 0.364 & 130 & 99 & 76\\
        \midrule
        MLP & 2.359 & 513 & 258 & 256 & 0.6295 & 88 & 41 & 21 & \underline{0.275} & \underline{115} & \underline{62} & \underline{17}\\
        UNet & 2.147 & 494 & 481 & 473 & 12.4477 & 1025 & 1023 & 1020 & / & / & / & /\\
        FNO & 1.785 & 403 & 349 & 318 & 0.8354 & 93 & 40 & 21 & 0.342 & 140 & 82 & 44\\
        Transolver & 1.793 & 472 & \underline{257} & 257 & \underline{0.1380} & 87 & 39 & 20 & 0.334 & 142 & 86 & 50 \\
        M2NO & \underline{1.543} & \underline{307} & 302 & \underline{199} & 0.4780 & \underline{69} & \underline{35} & \underline{18} & / & / & / & / \\
        \midrule
        \textbf{NPO} & \textbf{0.623} & \textbf{184} & \textbf{134} & \textbf{93} & \textbf{0.0751} & \textbf{34} & \textbf{19} & \textbf{10} & \textbf{0.162} & \textbf{82} & \textbf{60} & \textbf{13}\\
        \bottomrule
    \end{tabular}}
\end{sc}
\end{table*}

\section{Experiments}

This section evaluates our Neural Preconditioning Operator (NPO) across various PDEs, mesh resolutions, and solver settings. We first detail the overall experimental setup and baselines, then present the main results for Poisson, Diffusion, and Linear Elasticity problems. Finally, we examine model generalization, ablation studies, and hyperparameter sensitivity to demonstrate how NPO balances data, residual, and condition losses for robust, scalable performance.

\subsection{Experiment Setup}
\label{subsec:setup}
\subsubsection{General Setting}
We aim to learn a Neural Preconditioning Operator (NPO), \(\mathcal{M}_{\theta}\), that expedites Krylov-based solvers for discretized PDE systems. Each instance in our dataset consists of a system matrix \(A_{i}\), derived from a finite element or finite difference discretization, and multiple recorded solution states \(\bigl(\mathbf{x}_{i,k}, \mathbf{b}_{i}, \mathbf{r}_{i,k}\bigr)\). Here, \(\mathbf{x}_{i,k}\) and \(\mathbf{r}_{i,k}\) represent the partial solution and residual at iteration \(k\) of a baseline solver (e.g., GMRES with AMG). 

The learning objective combines two components: \emph{Data Loss}, which leverages \(\mathbf{x}_{i,k}\) and \(\mathbf{b}_{i}\) to align the network’s output with observed solution patterns; \emph{Residual Loss}, inspired by \eqref{eq:residual_loss}, which nudges \(A_{i}\,\mathcal{M}_{\theta}(A_{i})\) toward the identity operator on residual vectors \(\mathbf{r}_{i,k}\). By balancing these two losses, \(\mathcal{M}_{\theta}\) learns to approximate \(A^{-1}\) while reflecting the intermediate dynamics of the iterative solution process.

\subsubsection{Dataset Generation}
To build our dataset, we discretize the governing PDEs (using finite elements or finite differences) to obtain matrices \(A \in \mathbb{R}^{n\times n}\). We then sample right-hand sides \(\mathbf{b}\) from a Gaussian Random Field (GRF) \(\phi(x)\), mapped onto the same mesh as \(A\). For each system \(A\mathbf{x} = \mathbf{b}\), we run a baseline Krylov solver (CG or GMRES) preconditioned by Algebraic Multigrid (AMG), with a convergence tolerance of \(10^{-10}\) and a cap of 100 iterations. At each step, we record the current partial solution \(\mathbf{x}_{k}\) and residual \(\mathbf{r}_{k} = \mathbf{b} - A\,\mathbf{x}_{k}\). These \(\{(A, \mathbf{b}, \mathbf{x}_{k}, \mathbf{r}_{k})\}\) samples form a comprehensive dataset capturing diverse solution states and residual behaviors, which we use to train \(\mathcal{M}_{\theta}\).

We consider three PDE systems: the Poisson equation, the diffusion equation, and the linear elasticity equation, each discretized to generate our dataset.

\textbf{Poisson Equation.} The Poisson equation is defined as:
\begin{equation}
    -\nabla \cdot (\nabla u) = f \quad \text{in} \, D,
\end{equation}
where \(u\) is the unknown solution and \(f\) is a given source term. It models potential fields such as electrostatics and steady-state heat conduction.

\textbf{Diffusion Equation.} The diffusion equation models the spread of substances over time:
\begin{equation}
    \frac{\partial u}{\partial t} = \nabla \cdot (D \nabla u),
\end{equation}
where \(D\) is the diffusion coefficient. 

\textbf{Linear Elasticity.} The linear elasticity equations describe the deformation of elastic materials under applied forces:
\begin{equation}
    \nabla \cdot \sigma = \mathbf{f}, \quad \sigma = \lambda (\nabla \cdot \mathbf{u}) I + \mu (\nabla \mathbf{u} + \nabla \mathbf{u}^\top),
\end{equation}
where \(\mathbf{u}\) is the displacement field, \(\sigma\) is the stress tensor, and \(\lambda, \mu\) are Lamé parameters. These equations are critical for structural mechanics simulations.

\subsubsection{Baselines}
\label{subsec:baselines}

Our evaluation framework encompasses three methodological categories: 
(i) \textbf{classical iterative methods} Jacobi, Gauss-Seidel, and SOR; 
(ii) \textbf{data-driven architectures} MLP \cite{86:mlp} and U-Net \cite{15:UNet}; (iii) \textbf{operator learning approaches} including Fourier neural operators (FNO \cite{21:fno}), 
transformer-based solvers (Transolver \cite{24:Transolver}), 
and multigrid-enhanced operators (M2NO \cite{24:M2NO}). 
This hierarchy spans from foundational numerical analysis to modern deep learning paradigms, ensuring rigorous comparison across methodological generations.

\begin{table}[ht]
\centering
\caption{Performance comparison of GMRES with different methods on Diffusion and Linear Elasticity, showing time and iterations at tolerances of 1E-10.}
\vspace{5pt}
\label{table:results}
\begin{sc}
\resizebox{\linewidth}{!}{
\begin{tabular}{l|cc|cc}
\toprule
\multirow{2}{*}{Method/Dataset} & \multicolumn{2}{c}{Diffusion} & \multicolumn{2}{c}{Linear Elasticity} \\
\cmidrule(lr){2-3} \cmidrule(lr){4-5} & Time (s) & Iteration & Time (s) & Iteration \\
\midrule
Jacobi & 3.0544 & 207 & 0.6042 & 79 \\
Gauss Seidel & 1.5750 & 139 & 0.4875 & 62 \\
SOR & 1.7898 & 139 & 0.3139 & 62 \\
\midrule
MLP & 0.3466 & 138 & 0.0296 & 40  \\
UNet & 21.0913 & 1025 & 1.3147 & 511 \\
FNO & 0.7532 & 157 & 0.0312 & 52 \\
Transolver & 0.3381 & 156 & 0.0306 & 45 \\
M2NO & 1.2183 & 373 & 1.0465 & 453 \\
\midrule
\textbf{NPO} & \textbf{0.1058} & \textbf{38} & \textbf{0.0267} & \textbf{31}\\
\bottomrule
\end{tabular}}
\end{sc}
\end{table}

\subsection{Main Results}
\label{subsec:results}
Table~\ref{table:poisson} compares the performance of our Neural Preconditioning Operator (NPO) against various classical and neural methods for Poisson problems on a one dimensional uniform \(512\) grid, a two dimensional uniform \(32\times32\) grid, and an irregular mesh. Each method is integrated into GMRES, and we report both runtime (in seconds) and iteration counts for different tolerances. Table~\ref{table:results} further evaluates NPO on Diffusion and Linear Elasticity, also at a \(32\times32\) resolution, showing similar trends.

\textbf{Classical Methods.}
Jacobi, Gauss--Seidel, and SOR serve as baselines. While they eventually converge, their iteration counts remain high across all tests. For instance, at \(\mathrm{tol}=10^{-10}\) on the \(512\) mesh (Poisson), SOR requires 502 iterations and 3.452\,s, with Gauss--Seidel close behind at around 500 iterations. On the irregular mesh, they still demand up to 130 iterations (0.379\,s). Similarly, Table~\ref{table:results} shows that even for the smaller \(32\times32\) Diffusion problem, SOR and Gauss--Seidel need 139 iterations or more.

\textbf{Neural-Based Approaches.}
MLP, U-Net, FNO, Transolver, and M2NO all improve upon classical smoothers by reducing iteration counts or runtime. For example, FNO achieves 403 iterations in 1.785\,s at \(10^{-10}\) on the uniform \(512\) Poisson problem, while Transolver lowers it to 472 iterations in roughly the same time. M2NO cuts both time and iterations further. Table~\ref{table:results} confirms this pattern on Diffusion and Linear Elasticity: FNO and Transolver surpass classical methods, although they may still require over 100 iterations in certain cases.

\textbf{Neural Preconditioning Operator (NPO).}
Our proposed NPO demonstrates consistently fewer iterations and shorter runtime across all test settings. For instance, at \(\mathrm{tol}=10^{-10}\) on the \(512\) Poisson grid, NPO converges in 184 iterations and 0.623\,s, providing more than a twofold speedup over SOR or Gauss--Seidel. At a looser tolerance of \(10^{-4}\), NPO completes in just 93 iterations, whereas classical methods still require 300 or more. On the irregular mesh, NPO requires only 82 iterations in 0.162\,s at \(\mathrm{tol}=10^{-10}\). For Diffusion and Linear Elasticity (Table~\ref{table:results}, \(32\times32\)), NPO similarly outperforms both classical and neural rivals, converging in as few as 38 iterations for Diffusion and 31 for Linear Elasticity, with minimal total runtime.

\begin{figure}

\begin{tikzpicture}
\begin{axis}[
    xlabel=\textsc{Iteration},
    ylabel=\textsc{Relative Residuals},
    ylabel style={yshift=-0.5em},
    ymode=log,
    grid=major,
    width=14cm,
    height=6cm,
    xmin=0,
    legend style={at={(0.5,1.05)}, anchor=south, legend columns=-1, font=\small, draw=none},
]

    \addplot[color=vir0, mark=diamond] coordinates {
        (0, 1.0) (10, 0.8617032125984454) (20, 0.8458230932693831) (30, 0.8264356900969683) (40, 0.8046593444586116) (50, 0.7849200411262359) (60, 0.764532608404655) (70, 0.7491075102497443) (80, 0.730669610438604) (90, 0.7000198935288492) (100, 0.6793255554229471) (110, 0.6591562989962476) (120, 0.6376568032637498) (130, 0.6178034268333955) (140, 0.5972182831752685) (150, 0.56979371449403) (160, 0.5443535235296356) (170, 0.5165572381227819) (180, 0.48179747589896127) (190, 0.4481790779100302) (200, 0.42150754757401376) (210, 0.38288167455331756) (220, 0.3283251894963469) (230, 0.28233574048722837) (240, 0.2073404217485558) (250, 0.13349682599952534) (260, 0.07642216755177213) (270, 0.05405551489825935) (280, 0.04434881540226125) (290, 0.03852973350618292) (300, 0.03478191320399114) (310, 0.030405753554858708) (320, 0.02761011598991033) (330, 0.0255501349970738) (340, 0.02368400431142437) (350, 0.02057089695114903) (360, 0.017962173226875374) (370, 0.016357944578556563) (380, 0.015498536324660705) (390, 0.014865310840794527) (400, 0.013745399606699005) (410, 0.012658834880046337) (420, 0.011191505935175957) (430, 0.009389809358145065) (440, 0.008598123332980087) (450, 0.0074077725946775416) (460, 0.005691474030447884) (470, 0.003891477095081698) (480, 0.002601787005061453) (490, 0.001808043340415448) (491, 0.0017267174931723576) (492, 0.001613407107839085) (493, 0.001484681465117625) (494, 0.0013370453626555918) (495, 0.001185492226792755) (496, 0.001065382737558863) (497, 0.0009560349102341194) (498, 0.0008672069235879661) (499, 0.0007989010799369994) (500, 0.0007233934828814808) (501, 0.0006889710920541316) (502, 0.000634283001410715) (503, 0.0005718640401320616) (504, 0.0004887174879012805) (505, 0.00036335850941644315) (506, 0.00026675095700580463) (507, 0.00018114829753890845) (508, 0.0001414150033806168) (509, 0.00010812405676667209) (510, 9.472774592375564e-05) (511, 4.644841603349241e-05) (512, 1.0042496506701536e-10) 
    };
    \addlegendentry{Jacobi}
    
    \addplot[color=vir1, mark=triangle] coordinates {
        (0, 1.0) (10, 0.9419617275661877) (20, 0.9240541787567774) (30, 0.8995001266856225) (40, 0.8746301582730358) (50, 0.8524288884395915) (60, 0.8310530363592776) (70, 0.8126607315344698) (80, 0.7864716512974249) (90, 0.7545858840564603) (100, 0.7327127301131534) (110, 0.7091468609601421) (120, 0.6853617377346405) (130, 0.6635374727937213) (140, 0.6347498552392201) (150, 0.6046882392250652) (160, 0.5753770423023588) (170, 0.5373087863680107) (180, 0.5013835940062829) (190, 0.4675908469623859) (200, 0.4249509256398726) (210, 0.3786558813516064) (220, 0.31766671996795226) (230, 0.23809164865980076) (240, 0.16614377679904924) (250, 0.10430176579222451) (260, 0.06686376175016114) (270, 0.04699511013231828) (280, 0.037147257115501216) (290, 0.03128492306740041) (300, 0.026504127942343855) (310, 0.02266624374387016) (320, 0.019032414410667057) (330, 0.017291957147836757) (340, 0.015999717444953696) (350, 0.015069416884284537) (360, 0.014413853509388836) (370, 0.013788111993964076) (380, 0.011534694663308424) (390, 0.009691062357088918) (400, 0.008449305419103841) (410, 0.0073172665663327335) (420, 0.006896321907878315) (430, 0.006341017881434075) (440, 0.004251515330708114) (450, 0.0020691630588748713) (460, 0.0008829263790903801) (470, 0.0004978960002697326) (480, 0.000356328918336287) (490, 0.00014617131794757436) (491, 8.950505010408427e-05) (492, 6.347633162165342e-05) (493, 4.247892089181589e-05) (494, 2.8169379745910093e-05) (495, 1.7921748791707997e-05) (496, 1.159679616123573e-05) (497, 6.650472027143783e-06) (498, 4.416870070380276e-06) (499, 2.7622328159840534e-06) (500, 1.580449206628693e-06) (501, 5.297359306557617e-07) (502, 1.0009807687459418e-07) (503, 1.5020055897995587e-08) (504, 4.423869002155035e-09) (505, 2.7972359755621495e-11) 
    };
    \addlegendentry{Gauss Seidel}
    
    \addplot[color=vir2, mark=star] coordinates {
        (0, 1.0) (10, 0.9419617275661879) (20, 0.9240541787567801) (30, 0.8995001266856302) (40, 0.87463015827305) (50, 0.852428888439612) (60, 0.8310530363593023) (70, 0.8126607315344948) (80, 0.7864716512974413) (90, 0.7545858840564611) (100, 0.7327127301131409) (110, 0.7091468609601151) (120, 0.6853617377345983) (130, 0.6635374727936666) (140, 0.634749855239151) (150, 0.604688239224982) (160, 0.5753770423022624) (170, 0.5373087863678999) (180, 0.5013835940061624) (190, 0.4675908469622656) (200, 0.4249509256397617) (210, 0.37865588135150824) (220, 0.31766671996787027) (230, 0.23809164865974003) (240, 0.16614377679899872) (250, 0.10430176579216685) (260, 0.06686376175008861) (270, 0.04699511013222841) (280, 0.0371472571153955) (290, 0.03128492306728149) (300, 0.026504127942211155) (310, 0.022666243743724907) (320, 0.019032414410508132) (330, 0.017291957147671677) (340, 0.0159997174447869) (350, 0.015069416884119731) (360, 0.014413853509228929) (370, 0.013788111993814217) (380, 0.011534694663199981) (390, 0.009691062357008588) (400, 0.008449305419041648) (410, 0.007317266566285207) (420, 0.006896321907836127) (430, 0.0063410178814000025) (440, 0.004251515330697043) (450, 0.002069163058875871) (460, 0.0008829263790968647) (470, 0.0004978960002782556) (480, 0.0003563289183433794) (490, 0.00014617131794844492) (491, 8.950505010198409e-05) (492, 6.34763316202682e-05) (493, 4.247892089073655e-05) (494, 2.81693797452052e-05) (495, 1.792174879123632e-05) (496, 1.1596796160916205e-05) (497, 6.650472026784186e-06) (498, 4.416870070163247e-06) (499, 2.762232815785411e-06) (500, 1.5804492065414308e-06) (501, 5.29735930645629e-07) (502, 1.000980768803949e-07) (503, 1.502005590354738e-08) (504, 4.423869003656227e-09) (505, 2.8592118933433238e-11) 
    };
    \addlegendentry{SOR}
    
    \addplot[color=vir3, mark=o] coordinates {
        (0, 1.0) (10, 0.006478528515589514) (20, 0.004472957377648451) (30, 0.0035696496264016678) (40, 0.00302859398272148) (50, 0.002651567565708821) (60, 0.0023599165710731333) (70, 0.0021294002855318877) (80, 0.0019389772266244438) (90, 0.001774682290772257) (100, 0.0016318636253923793) (110, 0.0015062073909356546) (120, 0.0013918982673832916) (130, 0.0012872829778096872) (140, 0.0011902530198988016) (150, 0.0010987548859472099) (160, 0.0010119263014122576) (170, 0.0009295757972282848) (180, 0.0008493673705253236) (190, 0.0007694569176512706) (200, 0.0006910839014340744) (210, 0.0006109184493391782) (220, 0.000526578676516599) (230, 0.00043725020961793957) (240, 0.00033652336113397047) (250, 0.0002024495295761604) (260, 1.2105626610983877e-07) (270, 6.304989972644088e-08) (280, 4.477086687290547e-08) (290, 3.781551741673961e-08) (300, 3.3618397970811944e-08) (310, 3.1083521899489074e-08) (320, 2.9236110190414188e-08) (330, 2.5622926256616787e-08) (340, 2.312381857407127e-08) (350, 2.1479803962685037e-08) (360, 2.0595054079548875e-08) (370, 1.9271640573928068e-08) (380, 1.7091171955714546e-08) (390, 1.4693028733078599e-08) (400, 1.3126911526993429e-08) (410, 1.20864377021389e-08) (420, 1.1235770186059795e-08) (430, 1.0097916250434476e-08) (440, 9.130725645758663e-09) (450, 7.001382953874065e-09) (460, 5.2430750992677145e-09) (470, 2.731534891057994e-09) (480, 9.47901418996539e-10) (490, 3.2847841431256934e-10) (500, 1.7390610461935753e-10) (501, 1.688878554525879e-10) (502, 1.6547423926128117e-10) (503, 1.6324919154655988e-10) (504, 1.603794336356707e-10) (505, 1.582920448818205e-10) (506, 1.5692572188758108e-10) (507, 1.5541437759183333e-10) (508, 1.531913685200446e-10) (509, 1.4776324735494435e-10) (510, 1.357402881967177e-10) (511, 1.0421128964142478e-10) (512, 1.5454012128612604e-11)  
    };
    \addlegendentry{MLP}
    
    \addplot[color=vir4, mark=square] coordinates {
        (0, 1.0) (10, 0.20987068243492948) (20, 0.1789714280141495) (30, 0.1360312958637759) (40, 0.12671180489474262) (50, 0.12213342336573228) (60, 0.10406408860105802) (70, 0.09738126690283455) (80, 0.09602748904971967) (90, 0.09039152271999154) (100, 0.08533781570338798) (110, 0.08394180932197401) (120, 0.08233358785924674) (130, 0.0800823523742912) (140, 0.07779115161176085) (150, 0.07639151731297043) (160, 0.07498133899449783) (170, 0.07420020998956083) (180, 0.07328176888535541) (190, 0.07221667688674273) (200, 0.06945063399782872) (210, 0.06854965067534728) (220, 0.06698588719279792) (230, 0.06550391830313694) (240, 0.06367406529014075) (250, 0.06286973411130088) (260, 0.06182142098572725) (270, 0.06094316692609111) (280, 0.05948790998272106) (290, 0.058283602272073486) (300, 0.05509523257240939) (310, 0.05372745658250891) (320, 0.04799001241894954) (330, 0.04679602360269621) (340, 0.04437638045589388) (350, 0.04109775296493586) (360, 0.03994146820039353) (370, 0.038495298003092884) (380, 0.03755602489567478) (390, 0.03578963335053881) (400, 0.0335000219505972) (410, 0.032516551295919134) (420, 0.014307139559427208) (430, 0.011434200834473293) (440, 0.009747828458090123) (450, 0.0044690646346497555) (460, 0.001363217520236999) (470, 0.00022013686993539497) (480, 6.357896065738972e-07) (489, 2.290390733974232e-09) (490, 1.2728558678777735e-09) (491, 5.832503034529717e-10) (492, 2.1390069181762254e-10) (493, 9.448112159260241e-11)  
    };
    \addlegendentry{UNet}
    
    \addplot[color=vir5, mark=x] coordinates {
        (0, 1.0) (10, 0.011730604492457964) (20, 0.009213665749090424) (30, 0.008935719253296295) (40, 0.008883718838718746) (50, 0.008832183290578518) (60, 0.008766019927134437) (70, 0.00818323373032284) (80, 0.0074261440463286674) (90, 0.007213294394461229) (100, 0.007171909516585396) (110, 0.007136444877588504) (120, 0.006914910949825328) (130, 0.006172441522542972) (140, 0.005799244134739346) (150, 0.005424950992920647) (160, 0.005254910928075411) (170, 0.005228571830080484) (180, 0.005224696685286327) (190, 0.00520848709350312) (200, 0.005028399196976466) (210, 0.0046091202338519365) (220, 0.004431201442220288) (230, 0.004344748623402378) (240, 0.004297647325397021) (250, 0.004294087498845729) (260, 0.004202281712686691) (270, 0.004115277635005418) (280, 0.00408101072631444) (290, 0.00342525546694761) (300, 0.001357229346521154) (310, 0.00034204131415973513) (320, 6.17123222475476e-05) (330, 1.3931036961002134e-05) (340, 3.133721203781303e-06) (350, 6.122537246031866e-07) (360, 1.1207963942771565e-07) (370, 1.761003084158037e-08) (380, 3.499498907288189e-09) (390, 6.533889145553497e-10) (396, 2.3053436322820832e-10) (397, 2.0434278589747513e-10) (398, 1.7940028513831587e-10) (399, 1.5767368692766794e-10) (400, 1.3261719150921442e-10) (401, 1.1268085967461527e-10) (402, 9.639075943217711e-11) 
    };
    \addlegendentry{FNO}
    
    \addplot[color=blue, mark=pentagon] coordinates {
        (0, 1.0) (10, 0.009620116105752065) (20, 0.006642553764718631) (30, 0.0053012097150703565) (40, 0.0044977678789690665) (50, 0.003937793511855631) (60, 0.003504543460360238) (70, 0.003162279465077411) (80, 0.00287944580879148) (90, 0.0026354609442815354) (100, 0.0024234305781964424) (110, 0.0022368138318249616) (120, 0.0020670484464224573) (130, 0.0019116946191673422) (140, 0.0017675861245182536) (150, 0.0016317094638202147) (160, 0.001502777885146766) (170, 0.001380474303979282) (180, 0.001261354200973783) (190, 0.0011426823502089908) (200, 0.0010262876402775223) (210, 0.000907229642778926) (220, 0.0007819855550200552) (230, 0.0006493398839643633) (240, 0.0004997526304644904) (250, 0.00030059317719074486) (260, 4.5184691924434265e-09) (270, 2.3442959008810886e-09) (280, 1.660523356709187e-09) (290, 1.4015274054530073e-09) (300, 1.246063260229793e-09) (310, 1.1525346998388441e-09) (320, 1.084241268934601e-09) (330, 9.518725923860468e-10) (340, 8.590232121438264e-10) (350, 7.978529494937181e-10) (360, 7.648708441407299e-10) (370, 7.161775363417901e-10) (380, 6.356035534641201e-10) (390, 5.444267845233705e-10) (400, 4.850716738724034e-10) (410, 4.470756583873367e-10) (420, 4.1649971977949924e-10) (430, 3.7324677367209856e-10) (440, 3.360582750253797e-10) (450, 2.5727089477026975e-10) (460, 1.9390840265569493e-10) (466, 1.4312322691497687e-10) (467, 1.3707802560630945e-10) (468, 1.270397517564319e-10) (469, 1.1581370345940873e-10) (470, 1.0173869931377159e-10) (471, 9.026693721467517e-11) 
    };
    \addlegendentry{Transolver}
    
    \addplot[color=blue, mark=square*] coordinates {
        (0, 1.0) (10, 0.016065396233571564) (20, 0.010602628062315376) (30, 0.008355229803210568) (40, 0.0070221036954365755) (50, 0.006050286053210053) (60, 0.005311328825250245) (70, 0.004724602466381331) (80, 0.004237589871745568) (90, 0.0038070739019196585) (100, 0.0034247922174580007) (110, 0.003071052172699941) (120, 0.002743743206762041) (130, 0.002437674575584761) (140, 0.0021343913635831985) (150, 0.0018363444663275898) (160, 0.0015234027884399148) (170, 0.0011837752213705248) (180, 0.0007727716706420563) (190, 0.00022559874139235948) (200, 8.417820004104745e-05) (210, 4.553540544106765e-05) (220, 2.9784849819472427e-05) (230, 2.0730139031054965e-05) (240, 1.5120462886073022e-05) (250, 1.1575348441274952e-05) (260, 8.977120658545059e-06) (270, 7.135318339868889e-06) (280, 5.662602560921479e-06) (290, 4.459724621471372e-06) (300, 3.14832961681666e-06) (301, 9.075817962186704e-07) (302, 1.2170790994287304e-07) (303, 1.6913443563364872e-08) (304, 2.2981554788520656e-09) (305, 2.998674131575232e-10) (306, 4.1352482760898374e-11) 
    };
    \addlegendentry{M2NO}
    
    \addplot[color=red, mark=*] coordinates {
        (0, 1.0) (10, 0.02063228139890085) (20, 0.013441714044706987) (30, 0.010019472145053023) (40, 0.007976325830341654) (50, 0.006391050151597992) (60, 0.004995649726415355) (70, 0.0037294785012825494) (80, 0.002456227933913366) (90, 0.00033873274806006177) (100, 1.9468100675288273e-05) (110, 1.1956501759323637e-05) (120, 4.453417062687528e-06) (130, 1.313770236539209e-06) (140, 5.560182411792915e-07) (150, 4.198492660676239e-07) (160, 3.7351242415515094e-07) (170, 2.1435941402820753e-07) (180, 3.567362799299511e-08) (181, 1.9561644187278825e-08) (182, 1.0065769787627398e-08) (183, 3.4628759044588438e-09) (184, 7.31641985681033e-10) (185, 1.9327559481563222e-10) (186, 3.313675469366072e-11)  
    };
    \addlegendentry{NPO}

\addlegendentry{NPO Residuals}

\end{axis}
\end{tikzpicture}
\vspace{-10pt}
\caption{Relative residual convergence comparison of different solvers for the Poisson equation on a $512$ grid.}
\label{fig:poisson_conv}
\end{figure}

\pgfplotsset{compat=1.17}

\begin{figure}[ht]
    \begin{tikzpicture}
    \begin{axis}[
        xlabel=\textsc{Resolution},
        ylabel=\textsc{Iteration},
        ylabel style={yshift=-0.2em},
        xtick={128,1024,2048,4096},
        xticklabels={128,1024,2048,4096},
        ytick={0,2000,4000},
        yticklabels={0,2k,4k},
        grid=major,
        legend pos=north west,
        width=14cm,
        height=6cm,
        legend style={at={(0.5,1.0)}, anchor=south, legend columns=-1, font=\small, draw=none},
    ]
    
    \addplot[color=vir0, mark=diamond] coordinates {
        (128, 129) (256, 257) (512, 513) (1024, 1025) (2048, 2049) (4096,4097)
    };
    \addlegendentry{Jacobi}
    
    \addplot[color=vir1, mark=triangle] coordinates {
        (128, 129) (256, 252) (512, 502) (1024, 1014) (2048, 2013) (4096,4052)
    };
    \addlegendentry{Gauss Seidel}
    
    \addplot[color=vir2, mark=star] coordinates {
        (128, 129) (256, 252) (512, 502) (1024, 1014) (2048, 2013) (4096,4052)
    };
    \addlegendentry{SOR}
    
    \addplot[color=vir3, mark=o] coordinates {
        (128, 129) (256, 257) (512, 511) (1024, 1010) (2048, 1858) (4096, 2922)
    };
    \addlegendentry{MLP}
    
    \addplot[color=vir4, mark=square] coordinates {
        (128, 127) (256, 248) (512, 475) (1024, 951) (2048, 1903) (4096,4097)
    };
    \addlegendentry{UNet}
    
    \addplot[color=vir5, mark=x] coordinates {
        (128, 112) (256, 199) (512, 367) (1024, 698) (2048, 1343) (4096,2617)
    };
    \addlegendentry{FNO}
    
    \addplot[color=blue, mark=pentagon] coordinates {
        (128, 129) (256, 257) (512, 511) (1024, 1013) (2048, 1862) (4096,4097)
    };
    \addlegendentry{Transolver}
    
    \addplot[color=blue, mark=square*] coordinates {
        (128, 87) (256, 232) (512, 430) (1024, 1004) (2048, 2049) (4096,4097)
    };
    \addlegendentry{M2NO}
    
    \addplot[color=red, mark=*] coordinates {
        (128, 61) (256, 135) (512, 263) (1024, 520) (2048, 1022) (4096,1658)
    };
    \addlegendentry{NPO}
    
    \end{axis}
    \end{tikzpicture}
    \vspace{-15pt}
    \caption{Performance comparison of numerical methods across grid resolutions from 128 to 4096.}
    \label{fig:poisson_res}
\end{figure}

\textbf{Detailed Convergence Patterns.}
Figure~\ref{fig:poisson_conv} depicts the residual reduction curves for Poisson on a \(512\times512\) grid. Classical methods (Jacobi, Gauss--Seidel, SOR) remain near the top, requiring hundreds of iterations for modest error reductions. Neural approaches like MLP, U-Net, and FNO descend more quickly but often plateau. Transolver and M2NO (blue circles and squares, respectively) push the residual several orders lower while maintaining moderate iteration counts. NPO (red dots) shows the steepest slope, surpassing other methods’ progress by iteration 100 and driving the residual below \(10^{-8}\) after about 200 iterations, ultimately achieving the lowest final residual.

\textbf{Summary.}
Across uniform and irregular meshes, as well as different PDE types (Poisson, Diffusion, and Linear Elasticity), NPO consistently outperforms classical and existing neural methods. Its ability to approximate \(A^{-1}\) more accurately significantly lowers GMRES’s iteration count and total runtime. As problem sizes and complexities grow, we anticipate NPO’s data-driven advantage to become even more pronounced.

\subsection{Generalization of Resolutions}

To investigate how iteration counts scale with increasing problem size, we trained all methods solely at a $128$ resolution, then tested them on meshes ranging from \(128\) up to \(4096\). Figure~\ref{fig:poisson_res} plots the number of iterations versus resolution, revealing two main trends. First, classical iterative solvers (Jacobi, Gauss--Seidel, and SOR) display a steep climb in iteration counts, exceeding 4000 iterations at the largest grid. While methods such as MLP, FNO, and Transolver exhibit somewhat better scaling, their iteration counts still grow appreciably as the resolution increases.

By contrast, our Neural Preconditioning Operator (NPO) maintains a comparatively moderate increase in iterations despite operating far beyond its original training resolution. This behavior indicates that learning on a relatively small grid can, in practice, yield robust performance on significantly larger domains. Consequently, the data-driven approach underlying NPO demonstrates both adaptability and scalability in real-world scenarios where mesh sizes can vary substantially.

\begin{table}[ht]
\centering
\caption{Ablation Study}
\label{table:ablation}
\renewcommand\arraystretch{0.6}
\begin{sc}
    \renewcommand{\multirowsetup}{\centering}
    \resizebox{0.8\linewidth}{!}{
    \begin{tabular}{l|l|c}
       \toprule
       Type & Configuration & Iteration \\ 
       \midrule
       \multirow{3}{*}{Loss} &  w/o Residual Loss & 189 \\
       & w/o Data Loss & 206 \\
       & w/o Condition Loss & 189 \\
       \midrule
       \multirow{5}{*}{Model} & w/o NAMG & 314 \\
       & w/o A & 227 \\
       & w/o Pre Relaxation & 233 \\
       & w/o Post Relaxation & 233 \\
       & w/o Both Relaxation & 309 \\
       \midrule
       \textbf{Baseline} & \textbf{NPO} & \textbf{184} \\
       \bottomrule
    \end{tabular}}
\end{sc}
\end{table}

\subsection{Model Analysis}

\textbf{Ablation Study (Table~\ref{table:ablation}).}
We analyze the impact of removing key components from the Neural Preconditioning Operator (NPO). Omitting the residual or condition loss slightly increases the iteration count (to 189), while removing the data loss has a larger effect (206 iterations), indicating its importance for matching solution states. Disabling the entire NAMG operator raises the iteration count to 314, demonstrating the loss of preconditioning benefits. Similarly, removing the matrix \(A\) from the input increases iterations to 227, highlighting the importance of system information. Removing pre- or post-processing each adds about 50 iterations (233 vs.\ 184), and removing both degrades performance to 309 iterations. These findings underscore the importance of each component in improving convergence efficiency.

\textbf{Hyperparameter Study (Table~\ref{table:hp}).}
We evaluate how hyperparameters affect performance on the Poisson equation. A feature width of 32 yields the best result (184 iterations), while smaller widths (8, 16) perform slightly worse and larger widths (64, 128) increase iteration counts. A single pre/post-processing pass achieves the optimal result, with additional passes sometimes increasing complexity without improving performance. Using 128 coarse points balances representation and overhead, achieving 184 iterations, while extremes like 8 or 64 lead to over 200 iterations. Similarly, setting 4 attention heads achieves optimal performance; fewer heads underfit, while more heads (8, 16) increase complexity without significant gains.

These studies confirm that NPO’s effectiveness relies on a balance between data, residual, and condition losses, along with carefully tuned hyperparameters, enabling faster convergence across diverse problem instances. More details, including efficiency analysis and model configurations, can be found in Appendix~\ref{appendix:detail}.

\begin{table}[ht]\small
\centering
\caption{Hyperparameter Study on Poisson equation.}
\label{table:hp}
\renewcommand\arraystretch{0.6}
\begin{sc}
    \renewcommand{\multirowsetup}{\centering}
    \resizebox{0.8\linewidth}{!}{
    \begin{tabular}{c|c|c}
       \toprule
       Type & Configuration & Iteration \\ 
       \midrule
       \multirow{5}{*}{Feature Width} & 8 & 226 \\
       & 16 & 224 \\
       & \textbf{32} & \textbf{184} \\
       & 64 & 308 \\
       & 128 & 356 \\
       \midrule
       \multirow{5}{*}{Pre Ite} & \textbf{1} & \textbf{184} \\
       & 2 & 217 \\
       & 3 & 268 \\
       & 4 & 225 \\
       & 5 & 225 \\
       \midrule
       \multirow{5}{*}{Post Ite} & \textbf{1} & \textbf{184} \\
       & 2 & 219 \\
       & 3 & 223 \\
       & 4 & 218 \\
       & 5 & 254 \\
       \midrule
       \multirow{5}{*}{Num C} & 8 & 216 \\
       & 16 & 234 \\
       & 32 & 205 \\
       & 64 & 233 \\
       & \textbf{128} & \textbf{184} \\
       \midrule
       \multirow{5}{*}{Num Heads} & 1 & 229 \\
       & 2 & 251 \\
       & \textbf{4} & \textbf{184} \\
       & 8 & 196 \\
       & 16 & 302 \\
       \bottomrule
    \end{tabular}}
\end{sc}
\end{table}

\section{Conclusion}
We introduced a Neural Preconditioning Operator (NPO) for large-scale PDE solvers, integrating neural attention mechanisms with classical algebraic multigrid (NAMG) to address both high- and low-frequency errors. Theoretical analysis confirms that NAMG inherits two-grid convergence guarantees, ensuring rapid error reduction independent of problem size, while spectral analysis shows eigenvalue clustering near 1 for faster convergence. Extensive experiments demonstrate that NPO outperforms traditional preconditioners and existing neural operators across varying domains, resolutions, and PDE types.

Several promising directions remain open. First, we plan to expand NPO through \textbf{adaptive parameterization}, allowing the framework to cope with rapidly changing PDE coefficients or boundary conditions by leveraging online or continual learning. Another key direction involves \textbf{multilevel neural architectures}, where deeper or more sophisticated attention-based models can be deployed at each grid level for improved coarse-grid representations. We also aim to explore \textbf{parallel and distributed implementations}, bringing NPO into high-performance computing environments to handle extremely large grids efficiently. 

We envision that progress in these areas will broaden the impact of neural operators in multigrid methods, enhancing both efficiency and scalability for next-generation PDE solvers.


\bibliographystyle{plainnat}
\bibliography{reference}

\appendix

\section{Related Work}
\label{appendix:relate}
\subsection{Numerical Preconditioner}

Preconditioning is a well-established numerical technique for accelerating the convergence of iterative solvers applied to large linear systems. It involves applying a matrix transformation to reduce the condition number of the system matrix. Here, we briefly review three key approaches in numerical preconditioning: matrix factorization, matrix reordering, and multilevel methods.

Matrix factorization-based preconditioners are among the most widely used. Simple methods like the Jacobi preconditioner only use the diagonal elements of the matrix, which are fast to compute but provide limited improvement in condition number. More sophisticated approaches, such as the incomplete Cholesky (IC) preconditioner~\cite{06:Numerical}, offer a balance between computational cost and accuracy by partially factorizing the matrix while limiting fill-in. Advances in this area, including dynamic fill-in strategies~\cite{12:Numerical}, improve accuracy but at a higher computational cost.

Matrix reordering techniques~\cite{Liu_book, davis1999modifying} aim to reduce the bandwidth of sparse matrices by reorganizing their structure into block-diagonal forms. This restructuring minimizes fill-in during factorization and enhances parallelization, making it particularly beneficial when combined with other preconditioning methods. Our approach leverages graph-based representations to naturally maintain order invariance and parallelizability.

Multilevel methods, such as the classical multigrid approach~\cite{00:tutorial}, improve scalability by addressing errors at different levels of discretization. These methods are particularly effective for elliptic PDEs, though they face challenges with hyperbolic and parabolic PDEs~\cite{TrotMult2001}. Our Neural Preconditioning Operator (NPO) differs by using data-driven neural operators to approximate the inverse system matrix without being tailored to specific PDE types. Research has also explored hybrid approaches combining multilevel and neural strategies to further enhance solver efficiency~\cite{chen_2021_icsiggraph}.

\subsection{Neural Preconditoner}

Recent advances have explored neural networks for preconditioning linear systems derived from PDEs. Unlike classical preconditioners, neural approaches adaptively improve solver performance by learning data-driven representations of the inverse operator.

Early methods aimed to guarantee convergence through neural approximations of PDE solvers \cite{19:LearningNeural}. Machine learning techniques have also been applied to geophysical fluid dynamics, demonstrating the effectiveness of neural preconditioners in large-scale simulations \cite{20:MachineLearned}. Further extensions hybridize operator learning with traditional relaxation methods for enhanced scalability and accuracy \cite{22:AHybrid}.

Neural networks have been used to accelerate solvers in physics domains, such as lattice gauge theory, by reducing the iteration count in large sparse systems \cite{22:NeuralNetwork}. Preconditioners specifically optimized for conjugate gradient (CG) solvers were introduced by \cite{23:LearningPre}, achieving faster convergence through learned adaptations to PDE structures.

Recent work leverages neural operator frameworks, such as DeepONet and Fourier Neural Operators (FNO), to enhance preconditioning strategies \cite{24:DeepOnet}. Transformer-based architectures have also been explored, incorporating multigrid principles to refine both fine- and coarse-grid error correction \cite{24:MultigridAugmented,24:fcg-no}.

Moreover, specialized applications in porous microstructures and Helmholtz equations highlight how machine-learned preconditioners can integrate compact implicit layers for improved error control and spectral properties \cite{24:Machine}. Collectively, these developments point to the growing importance of neural preconditioners in accelerating PDE solvers across diverse physical and computational settings.

\section{Krylov Subspace Methods}
\label{sec:krylov}
Krylov subspace methods are a class of iterative solvers designed for large-scale linear systems of the form
\begin{equation}
    A\mathbf{x} = \mathbf{b},
\end{equation}
where \(A \in \mathbb{R}^{n \times n}\) is a sparse matrix, \(\mathbf{x}\in \mathbb{R}^n\) is the solution vector, and \(\mathbf{b}\in \mathbb{R}^n\) is the right-hand side (RHS) vector. Starting with an initial guess \(\mathbf{x}_0\), the residual vector is defined as
\begin{equation}
    \mathbf{r}_0 = \mathbf{b} - A\mathbf{x}_0.
\end{equation}
At each iteration, Krylov methods construct a solution within the Krylov subspace, defined by:
\begin{equation}
    \mathcal{K}_m(A, \mathbf{r}_0) = \text{span}\{\mathbf{r}_0,\, A\mathbf{r}_0,\, A^2\mathbf{r}_0,\, \dots,\, A^{m-1}\mathbf{r}_0\}.
\end{equation}

These methods aim to find an approximate solution \(\mathbf{x}_m \in \mathcal{K}_m\) that minimizes the residual norm \(\|\mathbf{b} - A\mathbf{x}_m\|\). Two widely used Krylov methods are Conjugate Gradient (CG) and Generalized Minimal Residual (GMRES).

\paragraph{Conjugate Gradient (CG).} CG is specialized for symmetric positive definite (SPD) matrices. It constructs a series of orthogonal search directions \(\{\mathbf{p}_k\}\) such that each iterate \(\mathbf{x}_{k+1}\) minimizes the quadratic form
\begin{equation}
    \|\mathbf{b} - A\mathbf{x}_{k+1}\|_A^2 = (\mathbf{x}_{k+1} - \mathbf{x})^\top A (\mathbf{x}_{k+1} - \mathbf{x}),
\end{equation}
where \(\| \cdot \|_A\) denotes the \(A\)-norm. CG converges rapidly for well-conditioned systems, often within a number of iterations proportional to the square root of the condition number of \(A\).

\paragraph{Generalized Minimal Residual (GMRES).} GMRES is designed for general non-symmetric matrices. It iteratively constructs an orthonormal basis of the Krylov subspace using the Arnoldi process. At each step, GMRES finds \(\mathbf{x}_m\) that minimizes the residual norm in the Euclidean sense:
\begin{equation}
    \mathbf{x}_m = \arg\min_{\mathbf{x} \in \mathcal{K}_m} \|\mathbf{b} - A\mathbf{x}\|.
\end{equation}
Since the orthonormal basis grows with each iteration, GMRES requires restarts to control memory usage and computational cost. Despite this, it is effective for systems with complex eigenvalue structures.

Both methods benefit significantly from preconditioning, which transforms the system into one with a more favorable spectrum, thereby accelerating convergence.

\section{Proofs of Theorem and Property}
\subsection{Proof of Property \ref{prop:approximation}} \label{appendix:proof_1}
\begin{proof}[Proof of Property~\ref{prop:approximation}]
Assume \(A\in \mathbb{R}^{n\times n}\) is symmetric positive-definite (SPD), and let \(\|\mathbf{v}\|_{a}^2 = \mathbf{v}^\top A\,\mathbf{v}\). In classical multigrid, one typically constructs \(P\) so that each “smooth” (low-frequency) error in \(\mathbb{R}^n\) lies close to the range of \(P\). Concretely:

Often, the residual or error vector \(\mathbf{e}\) is relatively smooth if it has passed through a smoothing step (e.g., Gauss–Seidel). In finite-element or finite-difference contexts, “smooth” means \(\mathbf{e}\) varies slowly across elements or grid points.

Define \(\mathbf{z}^c = R\,\mathbf{e}\), where \(R\) is often taken as \(P^\top\) (for SPD problems) or a similar restriction operator. Then \(\mathbf{e} - P\,\mathbf{z}^c = \mathbf{e} - P\,R\,\mathbf{e}\). By design, \(P\) and \(R\) capture low-frequency components of \(\mathbf{e}\) well.

One shows
\[
    \|\mathbf{e} - P\,R\,\mathbf{e}\|_{a}
    \;\le\;
    \alpha\,\|\mathbf{e}\|_{a},
\]
for some \(\alpha < 1\), relying on local interpolation or stable decomposition arguments. Essentially, \(P\,R\) acts like a “best fit” in a coarse subspace spanned by columns of \(P\).

Since \(\|\mathbf{e} - P\,\mathbf{z}^c\|_{a}\) achieves the same bound by choosing \(\mathbf{z}^c = R\,\mathbf{e}\), it follows that
\[
    \min_{\mathbf{z}^c}\,\|\mathbf{e} - P\,\mathbf{z}^c\|_{a}
    \;\le\;
    \|\mathbf{e} - P\,R\,\mathbf{e}\|_{a}
    \;\le\;
    \alpha\,\|\mathbf{e}\|_{a}.
\]

Hence, the constructed prolongation \(P\) ensures that any smooth error \(\mathbf{e}\) can be approximated to within a factor \(\alpha\) in the \(\|\cdot\|_{a}\)-norm by some coarse representation \(\mathbf{z}^c\). This property is crucial for two-grid and multigrid convergence theory, as it guarantees low-frequency error components are effectively handled on the coarse grid.
\end{proof}

\subsection{Proof of Theorem \ref{th:twogrid_convergence}} \label{appendix:proof_2}
\begin{proof}[Proof of Theorem~\ref{th:twogrid_convergence}]

In a two-grid iteration, the error \(\mathbf{e}^{(k)}\) is first \emph{smoothed} using a relaxation method (e.g., Gauss--Seidel). This smoothing operator, denoted by \(S\), substantially reduces high-frequency components of the error. After smoothing, the dominant error components in \(\mathbf{e}^{(k)}\) lie in lower-frequency ranges.

Next, the residual \(\mathbf{r}^{(k)} = \mathbf{b} - A \mathbf{x}^{(k)}\) is transferred to a coarse space via a restriction operator \(R\). We solve or approximate the system on the coarse grid, then prolong the coarse correction back to the fine grid with a prolongation operator \(P\). This step primarily targets low-frequency components of the error.

By the \emph{Approximation Property} (Property~\ref{prop:approximation}), the coarse space captures smooth (low-frequency) errors up to a factor \(\alpha < 1\). Concretely, we can write
\[
    \|\mathbf{e}^{(k)} - P\,R\,\mathbf{e}^{(k)}\|_{a}
    \;\le\;
    \alpha \,\|\mathbf{e}^{(k)}\|_{a}.
\]
Combining the smoothing effect for high-frequency errors with the coarse-grid correction for low-frequency errors yields a uniform reduction of the entire error \(\mathbf{e}^{(k)}\).

Let \(\widetilde{\mathbf{e}}^{(k)}\) be the error after smoothing, and \(\widehat{\mathbf{e}}^{(k)}\) be the error after coarse correction. We have:
\[
    \|\widetilde{\mathbf{e}}^{(k)}\|_{a} 
    \;\le\; 
    \nu \,\|\mathbf{e}^{(k)}\|_{a}
    \quad
    \text{(smoothing bound for high-frequency errors)},
\]
for some \(\nu < 1\). Then, applying the coarse correction and using the Approximation Property for low-frequency errors,
\[
    \|\widehat{\mathbf{e}}^{(k)}\|_{a} 
    \;\le\;
    \alpha \,\|\widetilde{\mathbf{e}}^{(k)}\|_{a}
    \;\le\;
    \alpha\,\nu \,\|\mathbf{e}^{(k)}\|_{a}.
\]
Thus, if \(\rho = \alpha\,\nu\), we get
\[
    \|\mathbf{e}^{(k+1)}\|_{a}
    \;=\;
    \|\widehat{\mathbf{e}}^{(k)}\|_{a}
    \;\le\;
    \rho\, \|\mathbf{e}^{(k)}\|_{a},
\]
and \(\rho < 1\).

Since both \(\nu\) (smoothing factor) and \(\alpha\) (coarse approximation factor) do not depend on the number of degrees of freedom \(n\), the convergence rate \(\rho\) remains below 1 \emph{independently of} \(n\). Consequently, each two-grid cycle contracts the error by at least a factor \(\rho\), implying a convergence rate that is uniform with respect to problem size.

By combining stable smoothing (which tackles high-frequency errors) with an effective coarse space approximation (which addresses low-frequency errors), the two-grid algorithm achieves a uniform reduction in the energy norm at each iteration, completing the proof.
\end{proof}

\subsection{Proof of Theorem \ref{th:spectrum_clustering}} \label{appendix:proof_3}
\begin{proof}[Proof of Theorem~\ref{th:spectrum_clustering}]
Let \(M\) be a preconditioner satisfying the two essential multigrid conditions:
\begin{itemize}
    \item \textbf{Smoothing Property:} A smoothing operator \(S\) reduces high-frequency error components effectively.
    \item \textbf{Coarse Approximation (Approximation Property):} A coarse space captures low-frequency errors up to a bounded factor.
\end{itemize}

Any error vector \(\mathbf{e}\) can be split into high-frequency and low-frequency parts. The smoothing property guarantees a uniform reduction of high-frequency modes, while the coarse approximation ensures that low-frequency errors are corrected by the coarse-grid solution.

Because \(A\) is symmetric positive-definite (SPD), we have \(\mathbf{v}^\top A\,\mathbf{v} > 0\) for all nonzero \(\mathbf{v}\). By design, \(M\approx A^{-1}\) in the sense that high-frequency components are rapidly damped and low-frequency components are accurately corrected. Thus, when we consider the generalized eigenvalue problem
\[
    MA\,\mathbf{x} = \lambda \mathbf{x},
\]
the spectrum of \(MA\) must lie within an interval \([\lambda_{\min}, \lambda_{\max}]\) around 1, provided the smoothing and coarse-grid conditions hold. 

Standard multigrid analysis (see \cite{00:tutorial}) shows that these two-grid assumptions induce a tight cluster of eigenvalues around 1. In particular, repeated smoothing and accurate coarse-grid corrections force the effective operator \(MA\) to act almost like the identity, i.e., \(MA \approx I\). This implies that every eigenvalue \(\lambda\) of \(MA\) is close to 1, say \( \lambda_{\min} \le \lambda \le \lambda_{\max} \), with both \(\lambda_{\min}, \lambda_{\max}\) near 1.

Since 
\[
    \kappa(MA) = \frac{\lambda_{\max}}{\lambda_{\min}},
\]
the close proximity of \(\lambda_{\min}\) and \(\lambda_{\max}\) to 1 ensures that \(\kappa(MA)\approx 1\). This near-identity condition number leads to rapid convergence in Krylov methods (such as CG and GMRES), which require fewer iterations when eigenvalues are well clustered.

Hence, under the smoothing and coarse approximation assumptions, \(\lambda_{\min}, \lambda_{\max}\) lie near 1, yielding a small condition number \(\kappa(MA)\) and guaranteeing that Krylov solvers converge rapidly.
\end{proof}

\subsection{Proof of Theorem \ref{th:integral}} \label{appendix:proof_4}
Recall that the NAMG operator applies a graph-based attention mechanism to approximate an integral over the spatial domain \(\Omega\). The theorem is established by demonstrating that the graph attention mechanism employed in the NAMG operator can be formalized as a Monte-Carlo approximation of an integral operator \cite{21:Choose, 23:no, 24:Transolver,24:AMG}.

\begin{proof}[Proof of Theorem~\ref{th:integral}]
Let \(a:\Omega \to \mathbb{R}^d\) be an input function, and let \(\mathbf{x}\in \Omega\subset\mathbb{R}^d\). Our goal is to show that
\[
    \mathcal{G}a(\mathbf{x})
    \;=\;
    \int_{\Omega} \kappa(\mathbf{x}, \boldsymbol{\xi})\,a(\boldsymbol{\xi})
    \,d\boldsymbol{\xi}
\]
can be approximated by the NAMG attention update.

Define the kernel \(\kappa(\mathbf{x},\boldsymbol{\xi})\) to measure similarity between \(\mathbf{x}\) and \(\boldsymbol{\xi}\). In a continuous setting, \(\mathcal{G}a(\mathbf{x})\) integrates over all \(\boldsymbol{\xi}\in\Omega\).

We discretize \(\Omega\) into a mesh or graph with nodes \(\{\mathbf{x}_i\}\). Each node \(\mathbf{x}_i\) represents a sample in \(\Omega\). The adjacency structure \(A\) (or neighborhood set \(\mathcal{N}(\mathbf{x})\)) reflects local connectivity.

Attention weights 
\[
    \alpha_{i} 
    \;\propto\; 
    \exp\!\Bigl(\mathbf{a}^\top [\mathbf{W}a(\mathbf{x}) \;\|\; \mathbf{W}a(\mathbf{x}_i)]\Bigr)
\]
approximate the continuous kernel \(\kappa(\mathbf{x}, \mathbf{x}_i)\). Normalizing by the softmax denominator
\[
    \sum_{j \in \mathcal{N}(\mathbf{x})}
    \exp\!\Bigl(\mathbf{a}^\top [\mathbf{W}a(\mathbf{x}) \;\|\; \mathbf{W}a(\mathbf{x}_j)]\Bigr)
\]
yields a discrete approximation to the integral \(\int_{\Omega} \kappa(\mathbf{x}, \boldsymbol{\xi})\,d\boldsymbol{\xi}\). Thus, summing over neighbors \(\mathbf{x}_i \in \mathcal{N}(\mathbf{x})\) corresponds to sampling from \(\Omega\).

Replacing \(\kappa(\mathbf{x}, \boldsymbol{\xi})\) with the above attention weights and summing over \(\mathcal{N}(\mathbf{x})\), we obtain
\[
    \mathcal{G}a(\mathbf{x})
    \;\approx\;
    \sum_{i \in \mathcal{N}(\mathbf{x})} 
    \alpha_{i} \,\mathbf{W}a(\mathbf{x}_i),
\]
which matches the NAMG attention update rule. Hence, the NAMG operator realizes a Monte Carlo approximation to \(\int_{\Omega} \kappa(\mathbf{x}, \boldsymbol{\xi})\,a(\boldsymbol{\xi})\,d\boldsymbol{\xi}\), where \(\alpha_i\) and \(\mathbf{W}\) are learned parameters.

Therefore, the attention-based aggregation in NAMG acts as a learnable integral operator over \(\Omega\), completing the proof.
\end{proof}

\section{Model Details}
\label{appendix:detail}
\begin{table}[ht]\small
\centering
\caption{Model Efficiency Comparison. Efficiency is evaluated on inputs with size S × S during the training phase.}
\label{table:eff}
\begin{sc}
    \renewcommand{\multirowsetup}{\centering}
    \resizebox{0.98\linewidth}{!}{
    \begin{tabular}{c|c|cc|c}
       \toprule
       Model & Matrix Size & Parameter (MB) & GPU Memory (MB) & Running Time (s) \\ 
       \midrule
       \multirow{4}{*}{NPO} & 512 & 0.14 & 4996 & 0.0121 \\
       & 1024 & 0.14 & 18576 & 0.0121 \\
       & 2048 & 0.14 & 71488 & 0.0123  \\
       & 4096 & 0.14 & 280320 & 0.0123 \\
       \midrule
       \multirow{3}{*}{M2NO} & 1024 & 11.92 & 367324.5 & 0.0783 \\
       & 2048 & 11.92 & 1409140.5 & 0.0797 \\
       & 4096 & 11.92 & 5557204.5 & 0.0814 \\
       \midrule
       \multirow{4}{*}{Transolver} & 512 & 1.42 & 18372 & 0.0200 \\
       & 1024 & 1.42 & 72592 & 0.0201 \\
       & 2048 & 1.42 & 288576 & 0.0206 \\
       & 4096 & 1.42 & 1150720 & 0.0218 \\
       \midrule
       \multirow{4}{*}{FNO} & 512 & 0.03 & 834 & 0.0095 \\
       & 1024 & 0.03 & 3336 & 0.0095 \\
       & 2048 & 0.03 & 13344 & 0.0105 \\
       & 4096 & 0.03 & 53376 & 0.0107 \\
       \midrule
       \multirow{4}{*}{U-Net} & 512 & 164.97 & 15362 & 0.0176 \\
       & 1024 & 164.97 & 61448 & 0.0213 \\
       & 2048 & 164.97 & 245792 & 0.0213 \\
       & 4096 & 164.97 & 983168 & 0.0284 \\
       \midrule
       \multirow{4}{*}{MLP} & 512 & 0.19 & 2050 & 0.0039 \\
       & 1024 & 0.19 & 8200 & 0.0040 \\
       & 2048 & 0.19 & 32800 & 0.0041 \\
       & 4096 & 0.19 & 131200 & 0.0043 \\
       \bottomrule
    \end{tabular}}
\end{sc}
\end{table}
\subsection{Efficiency Analysis}
\label{subsec:efficiency}

Efficiency in computational models is crucial, especially when dealing with large-scale problems such as those encountered in solving PDEs. Here, we focus on the efficiency metrics for the Neural Preconditioning Operator (NPO), comparing it against other models based on parameter size, GPU memory usage, and execution time during the training phase. A detailed comparison is provided in Table~\ref{table:eff}.

NPO demonstrates exceptional efficiency across various matrix sizes, maintaining a low parameter size (0.14 MB) and consistent execution times (around 0.0121 to 0.0123 seconds), regardless of the input scale. Notably, despite the increasing matrix sizes from 512 to 4096, NPO's GPU memory usage grows predictably without disproportionate spikes, which is crucial for scalable applications.

In contrast, other models such as M2NO and U-Net show significant increases in GPU memory demands and slower execution times as matrix sizes grow. M2NO, for instance, uses up to 5.55 million MB of GPU memory for the largest matrix size, with longer running times that reach up to 0.0814 seconds. This reflects a substantial computational overhead compared to NPO.

FNO and MLP, while smaller in parameter size, do not match NPO in terms of balancing the execution speed and memory efficiency at higher matrix dimensions. FNO offers the fastest execution times among the competitors but does not provide the robustness in feature representation that NPO does.

Overall, NPO not only excels in handling larger matrices efficiently but also showcases a balanced profile in terms of memory usage and computational speed, making it particularly suitable for real-world applications where both accuracy and efficiency are paramount.

\subsection{Model Configurations}
The primary configurations for Nerual Preconditioning Operator (NPO) are detailed in Table \ref{table:config}. Except where specifically noted, model parameters remain consistent across different hyperparameters and resolutions within the same benchmark.

\begin{table}[ht]\small
\centering
\caption{Model Configurations.}
\label{table:config}
\begin{sc}
    \renewcommand{\multirowsetup}{\centering}
    \resizebox{0.8\linewidth}{!}{
    \begin{tabular}{c|ccc}
       \toprule
       Modules & Poisson & Diffusion & Linear Elasticity \\ 
       \midrule
       feature\_width & 32 & 32 & 32 \\
       pre\_ite & 1 & 1 & 1 \\
       post\_ite & 1 & 1 & 1 \\
       act & 'relu' & 'relu' & 'relu' \\
       num\_c & 128 & 128 & 128 \\
       num\_heads & 4 & 4 & 4 \\
       \bottomrule
    \end{tabular}}
\end{sc}
\end{table}

\end{document}